\documentclass[preprint,12pt]{elsarticle}

\usepackage{graphicx}

\usepackage{amssymb}
\usepackage{amsmath}
\usepackage{amsthm}
\newtheorem{thm}{Theorem}[section]
\newtheorem{lem}{Lemma}[section]
\newtheorem{cory}{Corollary}[section]

\DeclareGraphicsExtensions{.eps,.jpg}

\begin{document}

\begin{frontmatter}



\title{Hamiltonian Paths in Two Classes of Grid Graphs}

\address[label1]{Corresponding author: fatemeh.keshavarz.2003@gmail.com}
\address[label2]{ ar\_bagheri@aut.ac.ir}

\author{Fatemeh Keshavarz-Kohjerdi \fnref{label1}}
\address{Department of Computer Engineering,\\ Islamic Azad University,
    North Tehran Branch, Tehran, Iran.\\}

\author{Alireza Bagheri \fnref{label2}}
\address{Department of Computer Engineering \& IT,\\ Amirkabir University of Technology, Tehran, Iran.}
\begin{abstract}
In this paper, we give the necessary and sufficient conditions for
the existence of Hamiltonian paths in $L-$alphabet and $C-$alphabet
grid graphs. We also present a linear-time algorithm for finding
Hamiltonian paths in these graphs.

\end{abstract}


\end{frontmatter}



\section{Introduction}\label{IntroSect}
A Hamiltonian path in a graph $G(V,E)$ is a simple path that
includes every vertex in $V$. The problem of deciding whether a
given graph $G$ has a Hamiltonian path is a well-known NP-complete
problem \cite{D:GT, GJ:CAI}. Rectangular grid graphs first appeared
in \cite{LM:HPOARC}, where Luccio and Mugnia tried to solve the
Hamiltonian path problem. The Hamiltonian path problem was studied
for grid graphs in \cite{IPS:HPIGG}, where the authors gave the
necessary and sufficient conditions for the existence of Hamiltonian
paths in rectangular grid graphs and proved that the problem for
general grid graphs is NP-complete. Also, the authors in
\cite{CT:HPOGG} presented sufficient conditions for a grid graph to
be Hamiltonian and proved that all finite grid graphs of positive
width have Hamiltonian line graphs.  Chen \textit{et al.}
\cite{CST:AFAFCHPIM} improved the algorithm of \cite{IPS:HPIGG} and
presented a parallel algorithm for the problem in mesh architecture.
In \cite{AER:MoSAG, AEB:SSAG}, Salman \textit{et al.} determined the
classes of alphabet graphs which contain Hamilton cycles. In this
paper, we obtain the necessary and sufficient conditions for a
$L-$alphabet and $C-$alphabet graphs to have Hamiltonian paths.
Also, we present a linear-time algorithm for finding Hamiltonian
paths in these graphs.

\section{Preliminaries}
In this section, we present some definitions and previously
established results on the Hamiltonian path problem
in grid graphs which appeared in \cite{CST:AFAFCHPIM, IPS:HPIGG, AEB:SSAG, AER:MoSAG}.\\
The \textit{two-dimensional integer grid} $G^\infty$ is an infinite
graph with the vertex set of all the points of the Euclidean plane
with integer coordinates. In this graph, there is an edge between
any two vertices of unit distance. For a vertex $v$ of this graph,
let $v_{x}$ and $v_{y}$ denote $x$ and $y$ coordinates of its
corresponding point. We color the vertices of the two-dimensional
integer grid by black and white colors. A vertex $\upsilon$ is
colored \textit{white} if $\upsilon_{x}+\upsilon_{y}$ is even, and
is colored \textit{black} otherwise. A \textit{grid graph} $G_{g}$
is a finite vertex-induced subgraph of the two-dimensional integer
grid. In a grid graph $G_{g}$, each vertex has degree at most four.
Clearly, there is no edge between any two vertices of the same
color. Therefore, $G_{g}$ is a bipartite graph. Note that any cycle
or path in a bipartite graph alternates between black and white
vertices. A \textit{rectangular grid graph} $R(m,n)$ (or $R$ for
short) is a grid graph whose vertex set is $V(R)= \{\upsilon \ |\ 1
\leq \upsilon_{x}\leq m, \ 1\leq \upsilon_{y}\leq n\}$. In the
figures, we assume that $(1,1)$ is the coordinates of the vertex in
the lower left corner. The size of $R(m,n)$ is defined to be
$m\times n$. $R(m,n)$ is called \textit{odd-sized} if $m\times n$ is
odd, and is called \textit{even-sized} otherwise. $R(m,n)$ is called a \textit{k-rectangle} if $n=k$.\\
The following lemma states a result about the Hamiltonicity of
even-sized rectangular graphs.
\begin{lem}
\label{Lemma:1} $[1]$ $R(m,n)$ has a Hamiltonian cycle if and only
if it is even-sized and $m,n>1$.
\end{lem}
Two different vertices $\upsilon$ and $\upsilon'$ in $R(m,n)$ are
called \textit{color-compatible} if either both $\upsilon$ and
$\upsilon'$ are white and $R(m,n)$ is odd-sized, or $\upsilon$ and
$\upsilon'$ have different colors and $R(m,n)$ is even-sized.
Without loss of generality, we assume $s_{x} \leq t_{x}$.\\ For
$m,n\geq3$, a $L-$alphabet graph $L(m,n)$ (or $L$ for short) and a
$C-$alphabet graph $C(m,n)$ (or $C$ for short) are subgraphs of $R(3
m -2,5 n -4)$ induced by $V(R)\backslash \{ V|v_x=m+1,\ldots, 3m-2$
and $v_y=n+1,\ldots, 5n-4\}$, and $V(R)\backslash \{
V|v_x=m+1,\ldots, 3m-2$ and $v_y=n+1,\ldots, 4n-4\}$, respectively.
These alphabet graphs are shown in Figure \ref{a} for $m=4$ and
$n=3$.\\ An alphabet graph is called \textit{odd-sized} if its
corresponding rectangular graph is
odd-sized, and is called \textit{even-sized} otherwise.\\
In the following by $L(m,n)$ we mean an $L-$alphabet grid graph
$L(m,n)$, and by $C(m,n)$ we mean an $C-$alphabet grid graph
$C(m,n)$. We use $P(A(m,n),s,t)$ to indicate the problem of finding
a Hamiltonian path between vertices $s$ and $t$ in grid graph
$A(m,n)$, and use $(A(m,n),s,t)$ to indicate the grid graph $A(m,n)$
with two specified distinct vertices $s$ and $t$ of it, where $A$ is
rectangular grid graph, $L-$alphabet graph or $C-$alphabet graph.
$(A(m,n),s,t)$ is \textit{Hamiltonian} if there is a Hamiltonian
path between $s$ and $t$ in $A(m,n)$. In this paper, since $even
\times odd$ $L-$alphabet ($C-$alphabet) graph and $odd\times even$
$L-$alphabet graph are isomorphic, then we only consider $even
\times odd$ $L-$alphabet ($C-$alphabet) graphs. \\An even-sized grid
graph contains the same number of black and white vertices. Hence,
the two end-vertices of any Hamiltonian path in the graph must have
different colors. Similarly, in an odd-sized grid graph the number
of white vertices is one more than the number of black vertices.
Therefore, the two end-vertices of any Hamiltonian path in such a
graph must be white. Hence, the color-compatibility of $s$ and $t$
is a necessary condition for $(R(m,n),s,t)$ to have Hamiltonian.
Furthermore, Itai \textit{et al.} \cite{IPS:HPIGG} showed that if
one of the following conditions hold, then $(R(m,n),s,t)$ is not
Hamiltonian:
\begin{itemize}
\item [(F1)] $R(m,n)$ is a 1-rectangle and either $s$ or $t$ is not a corner
vertex (Figure \ref{RecFig}(a)).
\item[(F2)] $R(m,n)$ is a 2-rectangle
and $(s,t)$ is a nonboundary edge, i.e. $(s,t)$ is an edge and it is
not on the outer face (Figure \ref{RecFig}(b)).
\item [(F3)] $R(m,n)$ is
isomorphic to a 3-rectangle $R'(m,n)$ such that $s$ and $t$ are
mapped to $s'$ and $t'$, and:
\begin{enumerate}
\item $m$ is even,
\item $s'$ is black, $t'$ is white,
\item  $s'_{y}=2$ and $s'_{x}<t'_{x}$ (Figure \ref{RecFig}(c)) or
$s'_{y}\neq 2$ and $s'_{x}< t'_{x}-1$ (Figure \ref{RecFig}(d)).
\end{enumerate}
\end{itemize}
A Hamiltonian path problem $P(R(m,n),s,t)$ is \textit{acceptable} if
$s$ and $t$ are color-compatible and $(R,s,t)$ does not satisfy any
of the conditions $(F1)$,
$(F2)$ and $(F3)$.\\
The following theorem has been proved in \cite{IPS:HPIGG}.
\begin{thm}
\label{Theorem:1} Let $R(m,n)$ be a rectangular graph and $s$ and
$t$ be two distinct vertices. Then $(R(m,n),s,t)$ is Hamiltonian if
and only if $P(R(m,n),s,t)$ is acceptable.
\end{thm}
\begin{figure}[thb]
 \centering
  \includegraphics[scale=1]{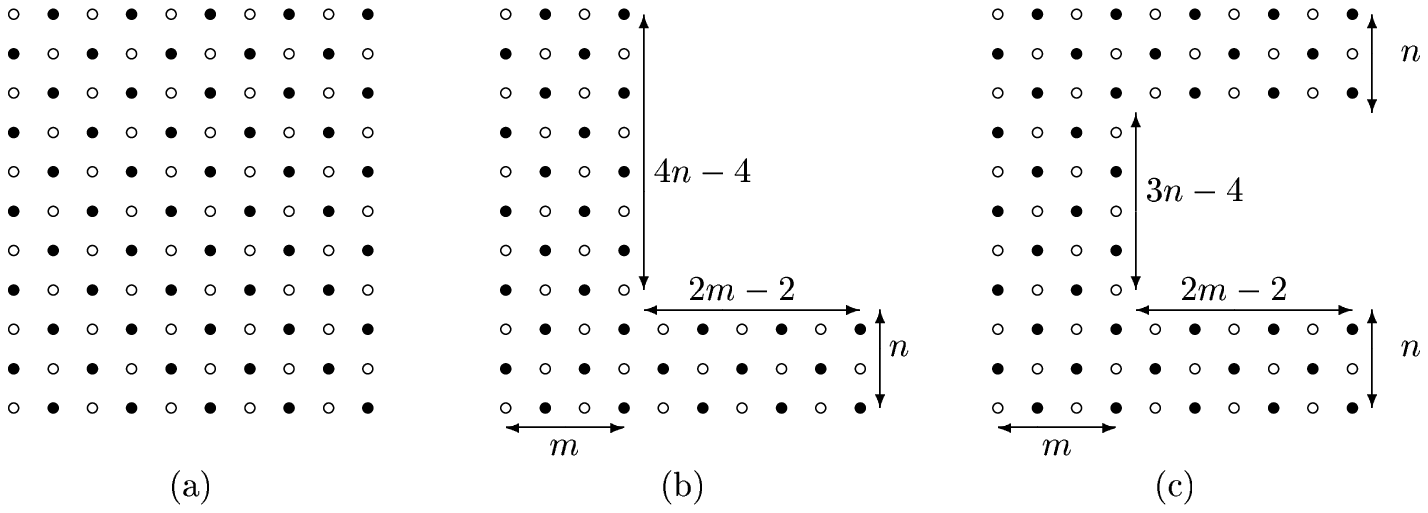}
  \caption[]%
  {\small (a) A rectangular grid graph $R(10,11)$, (b)  a $L-$alphabet grid graph
  $L(4,3)$, (c) a $C-$alphabet grid graph $C(4,3)$. }
  \label{a}
\end{figure}
\begin{figure}[thb]
  \centering
  \includegraphics[scale=1]{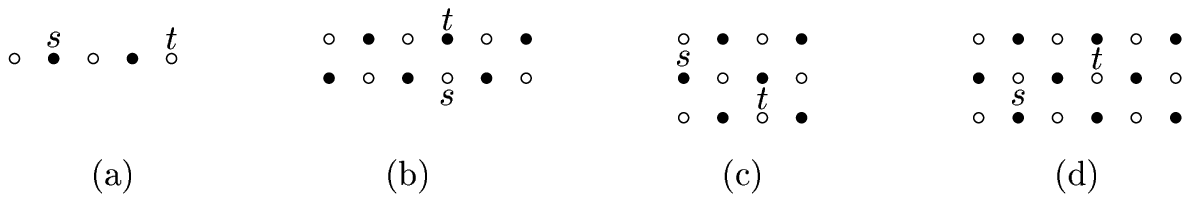}
  \caption[]%
  {\small Rectangular grid graphs in which there is no Hamiltonian path between $s$ and $t$.}
  \label{RecFig}
\end{figure}
\section{The Hamiltonian path in alphabet graphs}
In this section, we give sufficient and necessary conditions for the
existence of a Hamiltonian path in $L-$alphabet and $C-$alphabet
graphs. We also present an algorithm for finding a Hamiltonian path
between two given vertices of
these graphs. \\
A \textit{separation} of $L-$alphabet graph is a partition of $L$
into two vertex disjoint rectangular grid graphs $R_1$ and $R_2$,
i.e. $V(L)=V(R_{1})\cup V(R_{2})$, and $V(R_{1})\cap
V(R_2)=\emptyset$. A $separation$ of $C-$alphabet graph is a
partition of $C$ into a $L-$alphabet graph and a rectangular grid
graph.
\begin{lem} \label{Lemma:3} Let $A(m,n)$ be a $L-$alphabet or $C-$alphabet grid
graph and let $R$ be the smallest rectangular grid graph includes
$A$. If $(A(m,n),s,t)$ has a Hamiltonian path between $s$ and $t$,
then $(R(m,n),s,t)$ also has a Hamiltonian path.
\end{lem}
\begin{proof}
Since $R-A$ is a rectangular grid graph with even size of
$(2m-2)\times (4n-4)$ or $(2m-2)\times (3n-4)$, then by Lemma
\ref{Lemma:1} it has a Hamiltonian cycle. By combining the
Hamiltonian cycle and the Hamiltonian path of $(A(m,n),s,t)$ of [1],
a Hamiltonian path between $s$ and $t$ for $(R,s,t)$ is obtained.
\end{proof}
Combining Lemma \ref{Lemma:3} and Theorem \ref{Theorem:1} the
following corollary is trivial.
\begin{cory}
\label{Corollary:1} Let $A(m,n)$ be a $L-$alphabet or $C-$alphabet
grid graph and $R$ be its smallest including rectangular grid graph.
If $(A(m,n),s,t)$ has a Hamiltonian path between $s$ and $t$, then
$s$ and $t$ must be color-compatible in $R(m,n)$.
\end{cory}
Therefore, the color-compatibility of $s$ and $t$ is a necessary
condition for $(L(m,n),s,t)$ and $(C(m,n),s,t)$ to have Hamiltonian paths.\\
The length of a path in a grid graph means the number of vertices of
the path. In any grid graph, the length of any path between two
same-colored vertices is odd and the length of any path between two
different-colored vertices is even.
\begin{lem}
\label{Lemma:4} Let $L(m,n)$ be a $L-$alphabet grid graph and $s$
and $t$ be two given vertices of $L$. Let $R(2m-2,n)$ and
$R(m,5n-4)$ be a separation of $L(m,n)$. If $t_x>m+1$ and
$R(2m-2,n)$ satisfies condition $(F3)$, then $L(m,n)$ does not have
any Hamiltonian path between $s$ and $t$.
\end{lem}
\begin{proof}
Without loss of generality, let $s$ and $t$ be color-compatible.
Assume that $R(2m-2,n)$ satisfies condition $(F3)$. We show that
there is no Hamiltonian path in $L(m,n)$ between $s$ and $t$. Assume
to the contrary that $L(m,n)$ has a Hamiltonian path $P$. The
following cases are possible:\\
\indent Case 1. In this case, $t$ is
in $R(2m-2,n)$ and $s$ is not in $R(2m-2,n)$. Since $n=3$ there are
exactly three vertices $v,w$ and $u$ in $R(2m-2,n)$ which are
connected to $R(m,5n-4)$, as shown in Figure \ref{e}. The following
sub-cases are possible for the Hamiltonian path $P$.\\
\indent\indent Case 1.1. The Hamiltonian path $P$ of $L(m,n)$ that
starts from $s$ may enter to $R(2m-2,n)$ for the first time through
one of the vertices $v,w$ or $u$ and passes through all the vertices
of $R(2m-2,n)$ and end at $t$.
This case is
not possible because we assumed that $R(2m-2,n)$ satisfies $(F3)$
($t^{'}=t$ and $s^{'}=w$).\\
\indent\indent Case 1.2. The Hamiltonian path $P$ of $L(m,n)$ may
enter to $R(2m-2,n)$, passes through some vertices of it, then
leaves it and enter it again and passes through all the remaining
vertices of it and finally ends at $t$. In this case, two sub-paths
of $P$ which are in $R(2m-2,n)$ are called $P_1$ and $P_2$, $P_1$
from $v$ to $u$ ($v$ to $w$ or $u$ to $w$) and $P_2$ from $w$ to $t$
($u$ to $t$ or $v$ to $t$).
This
case is not also possible because the size of $P_1$ is odd (even)
and the size of $P_2$ is even (odd), then $|P_1+P_2|$ is odd while
$R(2m-2,n)$ is even, which is a contradiction.\\
\indent Case 2. In this case, $s$ and $t$ are in $R(2m-2,n)$. The
following cases may be considered:\\
\indent\indent Case 2.1. The Hamiltonian path $P$ of $L(m,n)$ starts
from $s$ passes through some vertices of $R(2m-2,n)$, leaves
$R(2m-2,n)$ at $v$ ($u$), then passes through all the vertices of
$R(m,5n-4)$ and reenter to $R(2m-2,n)$ at $w$ goes to $u (v)$ and
passes through all the remaining vertices of $R(2m-2,n)$ and ends at
$t$. In this case by connecting $v (u)$ to $w$ we obtain a
Hamiltonian path from $s$ to $t$ in $R(2m-2,n)$, which contradicts
to the assumption that $R(2m-2,n)$ satisfies $(F3)$, see Figure
\ref{l}(a).\\
\indent\indent Case 2.2. The Hamiltonian path $P$ of $L(m,n)$ starts
from $s$ leaves $R(2m-2,n)$ at $v (u)$, then passes through all the
vertices of $R(m,5n-4)$ and reenter to $R(2m-2,n)$ at $u(v)$ goes to
$w$ and passes through all the remaining vertices of $R(2m-2,n)$ and
ends at $t$. In this case, two parts of $P$ resides in $R(2m-2,n)$.
The part $P_1$ starts from $s$ ends at $v(u)$, and the part $P_2$
starts from $u(v)$ ends at $t$. The size of $P_1$ is even and the
size of $P_2$ is odd while the size of $R(2m-2,n)$ is even, which is
a contradiction, see
Figure \ref{l}(b).\\
\indent\indent Case 2.3. Another case that may imagine is that the
Hamiltonian path $P$ of $L(m,n)$ starts from $s$ leaves $R(2m-2,n)$
at $w$ and reenters $R(2m-2,n)$ at $v(u)$ and then goes to $t$. But
in this case vertex $u(v)$ can not be in $P$, which is a
contradiction, see Figure \ref{l}(c).\\ Thus the proof of Lemma
\ref{Lemma:4} is completed.
\end{proof}
\begin{figure}[thb]
  \centering
  \includegraphics[scale=1]{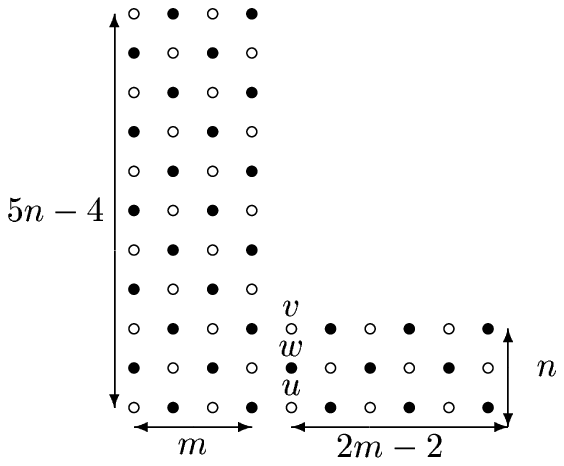}
  \caption[]%
  {\small}
  \label{e}
\end{figure}
\begin{figure}[thb]
  \centering
  \includegraphics[scale=1]{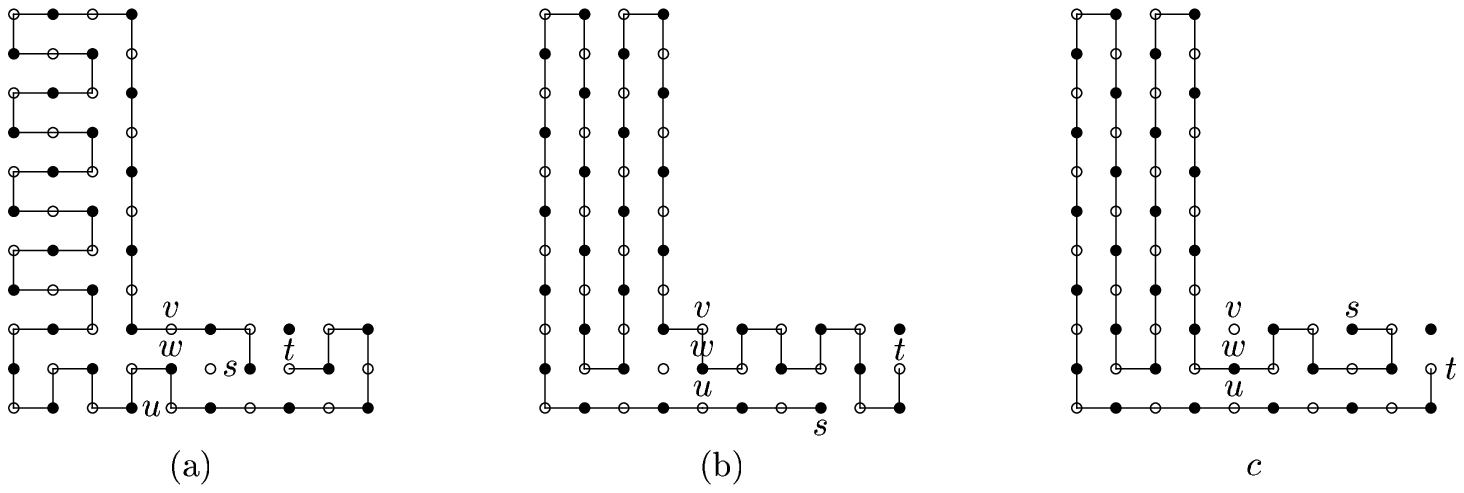}
  \caption[]%
  {\small  $L-$alphabet grid graphs in which there is no Hamiltonian path between $s$ and $t$.}
  \label{l}
\end{figure}
\begin{lem}
\label{Lemma:5} Assume that $C(m,n)$ is a $C-$alphabet grid graph
and $s$ and $t$ are two given vertices of $C(m,n)$. Let $L(m,n)$ and
$R(2m-2,n)$ be a separation of $C(m,n)$. If $L(m,n)$ does not have
Hamiltonian path, then $C(m,n)$ does not have Hamiltonian path
between $s$ and $t$.
\end{lem}
\begin{proof}
The proof is similar to the proof of Lemma \ref{Lemma:4}, for more
details see Figure \ref{v}.
\end{proof}
\begin{figure}[thb]
  \centering
  \includegraphics[scale=1]{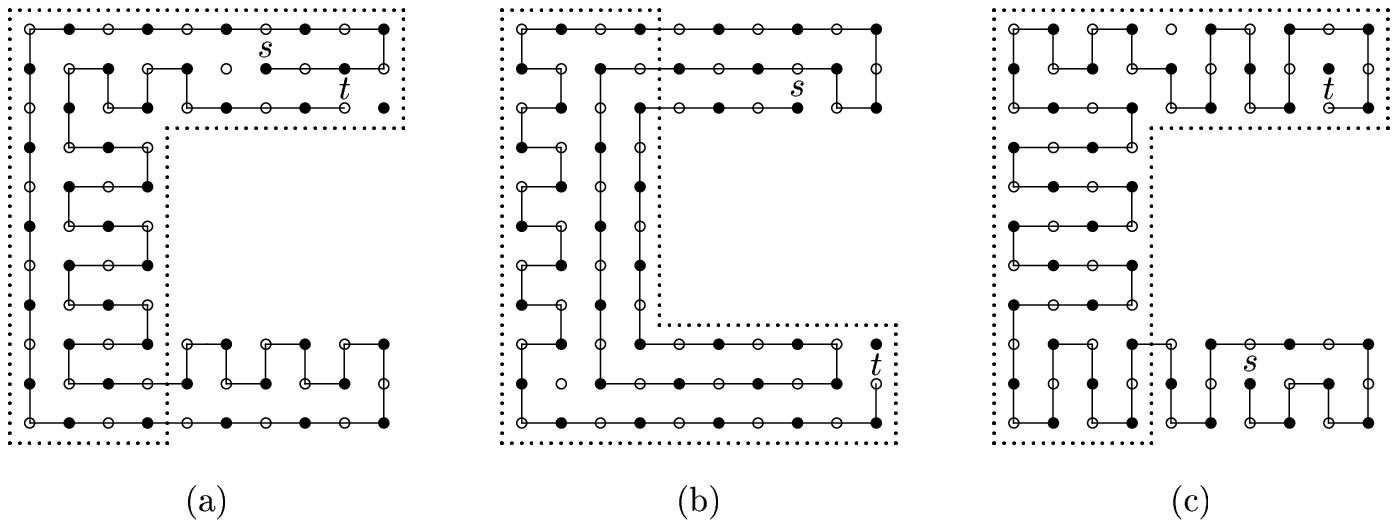}
  \caption[]%
  {\small  $C-$alphabet grid graphs in which there is no Hamiltonian path between $s$ and $t$.}
  \label{v}
\end{figure}
A Hamiltonian path problem $P(L(m,n),s,t)$ is \textit{acceptable} if
$s$ and $t$ are color-compatible and $R(2m-2,n)$ does not satisfy
the condition $(F3)$, and also $P(C(m,n),s,t)$ is
\textit{acceptable} if
$P(L(m,n),s,t)$ is acceptable, where $L(m,n)$ is a partition of $C(m,n)$.\\
 Now, we
are have shown that all acceptable Hamiltonian path problems have
solutions. Our algorithm is described in the following: \\ If the
given graph is $L-$alphabet, then it is divided into two rectangular
grid graphs and two Hamiltonian paths in them are found by the
algorithm of \cite{CST:AFAFCHPIM}.\\If the given graph is
$C-$alphabet, then it is divided into a $L-$alphabet graph and a
rectangular grid
graph, and the Hamiltonian path in $L-$alphabet graph is found as before.\\
In the following we discuss the details of this dividing and
merging.\\
A rectangular subgraph $S$ of $L-$alphabet or $C-$alphabet graph $A$
\textit{strips} a Hamiltonian path problem $P(A(m,n),s,t)$, if:
\begin{enumerate}
\item $S$ is even-sized.
\item $S$ and $A-S$ is a separation of $A$.
\item $s,t\in A-S$
\item $A-S$ is acceptable.
\end{enumerate}
\begin{lem}
\label{Lemma:9} Let $P(L(m,n),s,t)$ be an acceptable Hamiltonian
path problem, and $S$ strips it. If $L-S$ has a Hamiltonian path
between $s$ and $t$, then $(L(m,n),s,t)$ has a Hamiltonian path
between $s$ and $t$.
\end{lem}

\begin{proof}
Assume that $L-S$ has a Hamiltonian path $H$. $S$ is an even-sized
rectangular grid graph and it has Hamiltonian cycle by Lemma
\ref{Lemma:1}. There exists an edge $ab\in H$ such that $ab$ is on
the boundary of $L-S$ facing $S$. A Hamiltonian path for
$(L(m,n),s,t)$ can be obtained by merging $H$ and the Hamiltonian
cycle of $S$ as shown in Figure \ref{c}(a).
\end{proof}
Let $(R_p, R_q)$ be a separation of $L(m,n)$. If $s$ and $t$ are in
different partitions, then we consider two vertices $p$ and $q$ such
that $s,p \in R_p$, $q,t \in R_q$ and $(R_p, R_q)$ are acceptable.
Therefore, a Hamiltonian path for $(L(m,n),s,t)$ can be obtained by
connecting two vertices $p$ and $q$ as shown in Figure \ref{d}(a).

\begin{figure}[thb]
  \centering
  \includegraphics[scale=1]{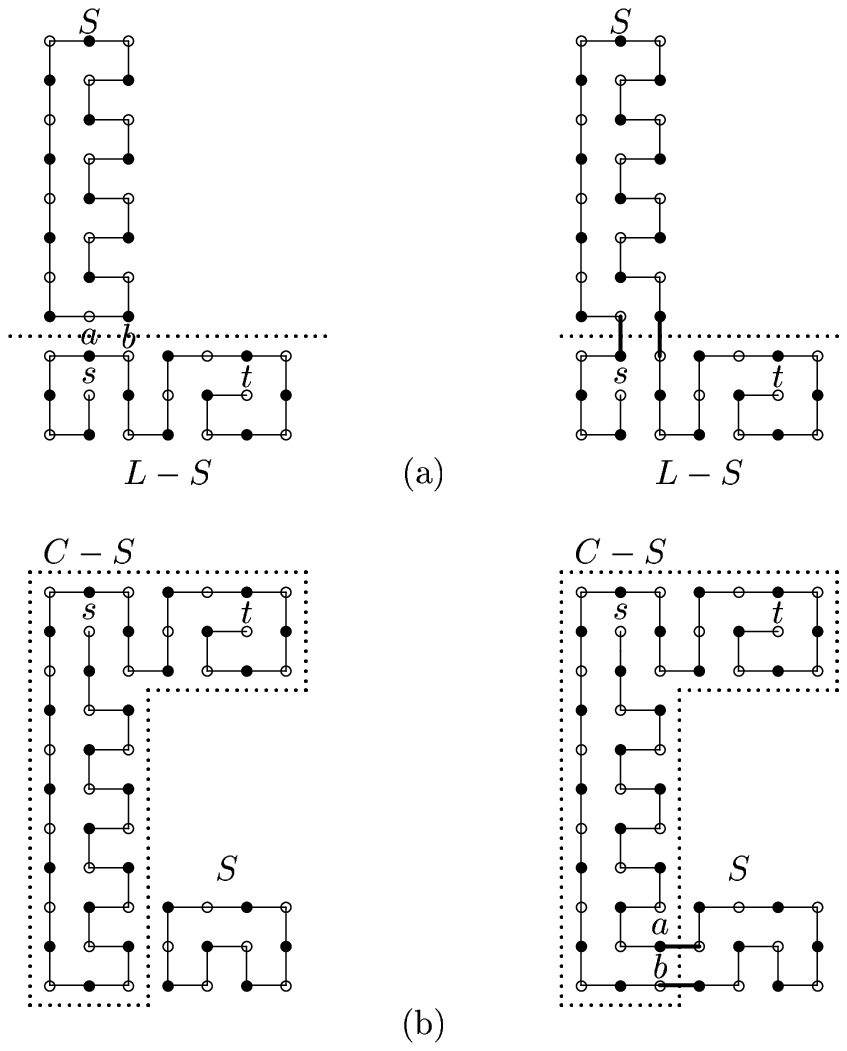}
  \caption[]%
  {\small  (a) A strip of $L(3,3)$, (b) A strip of $C(3,3)$.}
  \label{c}
\end{figure}
\begin{figure}[thb]
  \centering
  \includegraphics[scale=1]{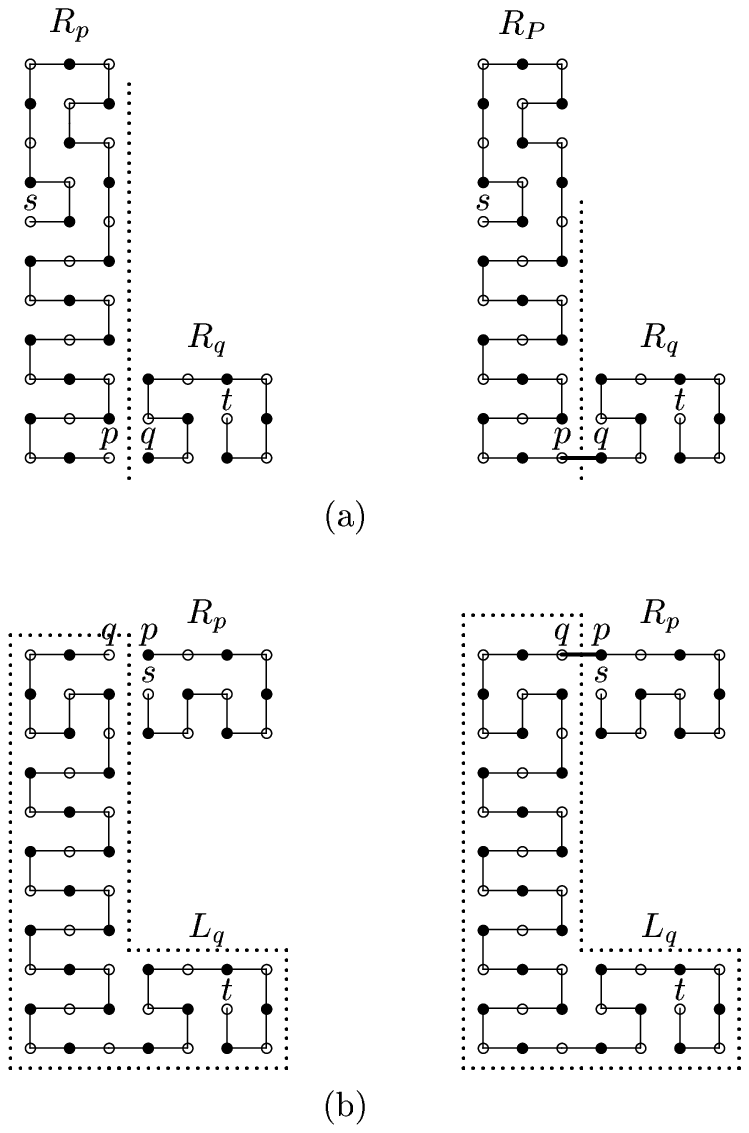}
  \caption[]%
  {\small (a) A split of $L(3,3)$, (b) A split of $C(3,3)$.}
  \label{d}
\end{figure}
\begin{lem}
\label{Lemma:10} Let $P(C(m,n),s,t)$ be an acceptable Hamiltonian
path problem, and $S$ strips it. If $C-S$ has a Hamiltonian path
between $s$ and $t$, then $(C(m,n),s,t)$ has a Hamiltonian path
between $s$ and $t$.
\end{lem}
\begin{proof}
The proof is similar to Lemma \ref{Lemma:9}. Notice that $C-S$ is a
$L-$alphabet grid graph, see Figure \ref{c}(b).
\end{proof}
Let $(R_p, L_q)$ be a separation of $C(m,n)$. If $s$ and $t$ are in
different partitions, then we consider two vertices $p$ and $q$ such
that $s,p \in R_p$, $q,t \in L_q$ and $(R_p,L_q)$ are acceptable.
Therefore, a Hamiltonian path for $(L(m,n),s,t)$ can be obtained by
connecting two vertices $p$ and
$q$ as shown in Figure \ref{d}(b). \\
From Corollary \ref{Corollary:1} and Lemmas \ref{Lemma:3},
\ref{Lemma:4}, 3.3, 3.4 and 3.5, the following theorem holds.
\begin{thm}
\label{Theorem:2} Let $A(m,n)$ be a $L-$alphabet or $C-$alphabet
graph, and let $s$ and $t$ be two distinct vertices of it. $A(m,n)$
has a Hamiltonian path if and only if $P(A(m,n),s,t)$ is acceptable.
\end{thm}
Theorem \ref{Theorem:2} provides the necessary and sufficient
conditions for the existence of Hamiltonian paths in $L-$alphabet
and $C-$alphabet grid graphs.
\begin{thm} \label{theorem:5}In $L-$alphabet and $C-$alphabet grid graphs,
a Hamiltonian path between any two vertices $s$ and $t$ can be found
in linear time.
\end{thm}
\begin{proof}
We divide the problem into two (or three) rectangular grid graphs in
O$(1)$. Then we solve the subproblems in linear-time and merge the
results in O(1) using the method proposed in [1].
\end{proof}

\section{Conclusion and future work}
\label{ConclusionSect} We presented a linear-time algorithm for
finding a Hamiltonian path in $L-$alphabet and $C-$alphabet grid
graphs between any two given vertices. Since the Hamiltonian path
problem is NP-complete in general grid graphs, it remains open if
the problem is polynomially solvable in solid grid graphs.

\end{document}